\documentclass[a4paper]{article} 

\usepackage[utf8]{inputenc}
\usepackage{lmodern}
\usepackage[T1]{fontenc}
\usepackage[mathscr]{eucal}
\usepackage[tbtags,fleqn]{amsmath}
\usepackage{amssymb}
\usepackage{amsthm}
\usepackage{authblk}
\usepackage[USenglish]{babel}
\usepackage[a4paper]{geometry}

\usepackage{mathpartir}
\usepackage{comment}

\allowdisplaybreaks[4]

\theoremstyle{plain}
\newtheorem{theorem}{Theorem}
\newtheorem{lemma}[theorem]{Lemma}
\newtheorem{corollary}[theorem]{Corollary}
\newtheorem{proposition}[theorem]{Proposition}
\theoremstyle{definition}
\newtheorem{definition}[theorem]{Definition}
\newtheorem{example}[theorem]{Example}
\theoremstyle{remark}
\newtheorem*{remark}{Remark}

\newcommand{\ML}{\mathsf{MATLANG}}
\newcommand{\inv}{\mathsf{inv}}
\newcommand{\eigen}{\mathsf{eigen}}

\newcommand{\transp}[1]{#1^*}
\newcommand{\SyntaxStyle}{\mathsf}
\newcommand{\LetIn}[3]{\SyntaxStyle{let}\ #1=#2\ \SyntaxStyle{in}\ #3}

\newcommand{\rowdom}{\mathbf{1}}
\DeclareMathOperator{\diag}{\mathsf{diag}}
\DeclareMathOperator{\Apply}{\mathsf{apply}}

\newcommand{\one}{\rowdom}

\newcommand{\dotmin}{\mathbin{\dot-}}
\newcommand{\M}{M}

\newcommand{\C}{\mathbf{C}}

\DeclareMathOperator{\var}{var}
\newcommand{\scm}{\mathcal{S}}

\newcommand{\bigstep}[3]{#2(#1)=#3}
\newcommand{\typed}[3]{#1\vdash#2:#3}

\DeclareMathOperator{\Sum}{\mathsf{sum}}
\newcommand{\Tb}{\mathbf{b}}
\newcommand{\Tn}{\mathbf{n}}
\newcommand{\dom}{\mathbf{dom}}

\newcommand{\Rel}{\mathit{Rel}}
\newcommand{\Mat}{\mathit{Mat}}
\newcommand{\Adj}{\mathit{Adj}}

\title{On the expressive power of query languages for matrices}
\date{}

\author[1]{Robert Brijder}
\author[2]{Floris Geerts}
\author[1]{Jan Van den Bussche}
\author[1]{Timmy Weerwag}
\affil[1]{Universiteit Hasselt}
\affil[2]{Universiteit Antwerpen}

\begin{document}

\maketitle

\begin{abstract}

We investigate the expressive power of $\ML$, a formal language
for matrix manipulation based on common matrix operations and
linear algebra.  The language can be extended with the operation $\inv$
of inverting a matrix.  In $\ML+\inv$ we can compute the
transitive closure of directed graphs, whereas we show that this
is not possible without inversion.  Indeed we show that the basic
language can be
simulated in the relational algebra with arithmetic operations,
grouping, and summation.  We also consider an operation $\eigen$
for diagonalizing a matrix, which is defined so that different
eigenvectors returned for a same eigenvalue are orthogonal.  We
show that $\inv$ can be expressed in $\ML+\eigen$.  We put forward
the open question whether there are boolean queries about
matrices, or generic queries about graphs, expressible in $\ML +
\eigen$ but not in $\ML+\inv$.  The evaluation problem for $\ML +
\eigen$ is shown to be complete for the complexity class $\exists \mathbf{R}$.

\end{abstract}

\section{Introduction}

Data scientists often use matrices to represent their data, as
opposed to using the relational data model.  These matrices are
then manipulated in programming languages such as R or
\textsf{MATLAB}.  These languages have common operations on
matrices built-in, notably matrix multiplication; matrix
transposition; elementwise operations on the entries of matrices;
solving nonsingular systems of
linear equations (matrix inversion); and diagonalization
(eigenvalues and eigenvectors).  Providing database support for
matrices and multidimensional arrays has been a long-standing
research topic \cite{rusu_survey}, originally geared towards
applications in scientific data management, and more recently
motivated by machine learning over big data
\cite{systemml,olteanu_regression,naughton_la,ngolteanu_learning}.

Database theory and finite model theory provide a rich picture of
the expressive power of query languages
\cite{ahv_book,kolaitis_expressivepower}.  In this paper we would
like to bring matrix languages into this picture.  There is a lot
of current interest in languages that combine matrix operations
with relational query languages or logics, both in database
systems \cite{hutchison} and in finite model theory
\cite{dawar_linearalgebra,dghl_rank,holm_phd}.  In the present
study, however, we focus on matrices alone.  Indeed, given their
popularity, we believe the expressive power of matrix
sublanguages also deserves to be understood in its own right.

The contents of this paper can be introduced as follows.  We
begin the paper by defining the language $\ML$ as an
analog for matrices of the relational algebra for relations.
This language is based on five elementary operations, namely,
the one-vector; turning a vector in a diagonal matrix; matrix
multiplication; matrix transposition; and pointwise function application.
We give examples showing that this basic language is capable
of expressing common matrix manipulations.  For example, the
Google matrix of any directed graph $G$ can be computed in $\ML$,
starting from the adjacency matrix of $G$.

Well-typedness and well-definedness notions of $\ML$ expressions are
captured via a simple data model for matrices. In analogy to
the relational model, a schema consists of a number of matrix
names, and an instance assigns matrices to the names.  Recall
that in a relational schema, a relation name is typed by a set of
attribute symbols.  In our case, a matrix name is typed by a pair
$\alpha \times \beta$, where $\alpha$ and $\beta$ are size symbols
that indicate, in a generic manner,
the number of rows and columns of the matrix.  

In Section~\ref{secsum} we show that our language can be simulated in
the relational algebra with aggregates
\cite{klug_agg,libkin_sql}, using a standard representation of
matrices as relations.  The only aggregate function that is
needed is summation.  In fact, $\ML$ is already subsumed by
aggregate logic with only three nonnumerical variables.
Conversely, $\ML$ can express all queries from graph databases
(binary relational structures) to binary relations that can be
expressed in first-order logic with three variables.  In
contrast, the four-variable query asking if the graph contains a
four-clique, is not expressible.

In Section~\ref{secinv} we extend $\ML$ with an operation for
inverting a matrix, and we show that the extended language
is strictly more expressive.  Indeed, the transitive
closure of binary relations becomes expressible.  The possibility
of reducing transitive closure to matrix inversion has been
pointed out by several researchers
\cite{laubner_phd,schwentick_reachdynfo,sato_ladatalog}.  We show
that the restricted setting of $\ML$ suffices for this reduction
to work.  That transitive closure is not expressible without
inversion, follows from the locality of relational algebra with
aggregates \cite{libkin_sql}.

Another prominent operation of linear algebra, with many
applications in data mining and graph analysis
\cite{hms_miningbook,ullman_mining}, is to return eigenvectors
and eigenvalues.  There are various ways to define this operator
formally. In Section~\ref{seceigen} we define the operation
$\eigen$ to return a basis of eigenvectors, in which eigenvectors
for a same eigenvalue are orthogonal.  We show that the
resulting language $\ML + \eigen$ can express inversion.
The argument is well known from linear algebra, but our result
shows that it can be carried out in $\ML$, once more attesting
that we have defined an adequate matrix language.  It is natural
to conjecture that $\ML + \eigen$ is actually strictly more
powerful than $\ML + \inv$ in expressing, say, boolean queries
about matrices. Proving this is an interesting open problem.

Finally, in Section~\ref{seceval} we look into the evaluation
problem for $\ML + \eigen$ expressions.  In practice, matrix
computations are performed using techniques from numerical
mathematics \cite{golub}.  It remains of foundational interest,
however, to know whether the evaluation of expressions is
effectively computable.  We need to define this problem with some
care, since we work with arbitrary complex numbers. Even if the
inputs are, say, 0-1 matrices, the outputs of the $\eigen$ operation
can be complex numbers.  Moreover, until now we have allowed
arbitrary pointwise functions, which we should restrict somehow
if we want to discuss computability.  Our approach is to restrict
pointwise functions to be semi-algebraic, i.e., definable over
the
real numbers.  We will observe that the input-output relation of
an expression $e$, applied to input matrices of given dimensions,
is definable in the existential theory of the real numbers, by a
formula of size polynomial in the size of $e$ and the given
dimensions.   This places natural decision versions of the
evaluation problem for $\ML + \eigen$ in the complexity class
$\exists \mathbf{R}$ (combined complexity).  We show moreover
that there exists a fixed expression (data complexity) for which
the evaluation problem is $\exists \mathbf{R}$-complete, even
restricted to input matrices with integer entries.  It also
follows that equivalence of expressions, over inputs of given
dimensions, is decidable.

\section{$\ML$}

We assume a sufficient supply of \emph{matrix variables}, which
serve to indicate the inputs to expressions in $\ML$. Variables
can also be introduced in let-constructs inside expressions.
The syntax of $\ML$ expressions is defined by the grammar:
\begin{align*}
  e &::= \M && \text{(matrix variable)} \\
    &\mid\quad \LetIn{\M}{e_1}{e_2} && \text{(local binding)} \\
    &\mid\quad \transp{e} && \text{(conjugate transpose)} \\
    &\mid\quad \rowdom(e) && \text{(one-vector)} \\
    &\mid\quad \diag(e) && \text{(diagonalization of a vector)} \\
    &\mid\quad e_1 \cdot e_2 && \text{(matrix multiplication)} \\
    &\mid\quad \Apply[f](e_1, \ldots, e_n) && \text{(pointwise
    application, $f \in \Omega$)}
\end{align*}
In the last rule, $f$ is the name of a function $f : \C^n \to
\C$, where $\C$ denotes the complex numbers.  Formally, the
syntax of $\ML$ is parameterized by a repertoire $\Omega$
of such functions, but for simplicity we will not reflect this in
the notation.

\begin{example} \label{exex}
Let $c \in \C$ be a constant; we also use $c$ as a name for
the constant function $c : \C \to \C : z \mapsto
c$.  Then $$ \LetIn{N}{\transp{\rowdom(M)}}{\Apply[c](\rowdom(N))}
$$ is an example of an expression.  At this point, this is a
purely syntactical example; we will see its semantics shortly.
The expression is actually equivalent to
$\Apply[c](\rowdom(\transp{\rowdom(M)}))$.  The
let-construct is useful to give names to intermediate results,
but is not essential for now.  It will become essential
later, when we enrich $\ML$ with the $\eigen$ operation.
\qed
\end{example}

In defining the semantics of the language, we begin by defining
the basic matrix operations.  Following practical matrix sublanguages
such as R or \textsf{MATLAB}, we will work throughout
with matrices over the complex numbers.  However, a real-number version
of the language could be defined as well.
\begin{description}
\item[Transpose:] If $A$ is a matrix then $\transp A$ is its
conjugate transpose.  So, if $A$ is an $m \times n$ matrix then
$\transp A$ is an $n \times m$ matrix and the entry $A^*_{i,j}$ is the
complex conjugate of the entry $A_{j,i}$.
\item[One-vector:] If $A$ is an $m \times n$ matrix then
$\one(A)$ is the $m \times 1$ column vector consisting of all
ones.
\item[Diag:] If $v$ is an $m \times 1$ column vector then
$\diag(v)$
is the $m \times m$ diagonal square matrix with $v$ on the diagonal and
zero everywhere else.
\item[Matrix multiplication:] If
$A$ is an $m \times n$ matrix and $B$ is an $n \times p$ matrix
then the well known matrix multiplication $A B$
is defined to be the $m \times p$ matrix
where $(AB)_{i,j} = \sum_{k=1}^n A_{i,k}B_{k,j}$.  In $\ML$ we
explicitly denote this as $A \cdot B$.
\item[Pointwise application:] If $A^{(1)},\dots,A^{(n)}$ are
matrices of the same dimensions $m \times p$, then
$\Apply[f](A^{(1)},\dots,A^{(n)})$ is the $m \times p$ matrix $C$
where $C_{i,j} = f(A^{(1)}_{i,j},\dots,A^{(n)}_{i,j})$.
\end{description}

\begin{figure}
$$
\begin{array}{l@{\hspace{2em}}l}
\displaystyle
\begin{pmatrix} 0 & 1+i \\ 2 & 3-i \\ 4 +4i & 5
\end{pmatrix}^* =
\begin{pmatrix} 0 & 2 & 4-4i  \\ 1-i & 3+i & 5 \end{pmatrix}
&
\displaystyle
\one\begin{pmatrix} 2 & 3 & 4 \\ 4 & 5 & 6 \end{pmatrix} =
\begin{pmatrix} 1 \\ 1 \end{pmatrix}
\\[4.5ex]
\begin{pmatrix} 1 & 2 \\ 3 & 4 \\ 5 &
6 \end{pmatrix} \cdot \begin{pmatrix} 6 & 5 & 4 & 3 \\ 2 & 1 & 0
  & -1 \end{pmatrix} = \begin{pmatrix} 10
   & 7 & 4 & 1 \\ 26 & 19 & 12 & 5 \\ 42 & 31 & 20 & 9
   \end{pmatrix}
&
\displaystyle
\diag \begin{pmatrix} 6 \\ 7 \end{pmatrix} =
\begin{pmatrix} 6 & 0 \\ 0 & 7 \end{pmatrix}
\end{array} $$
$$
\Apply[\dotmin](\begin{pmatrix} 1 & 1 & 1 \\ 0 & 1 & 1 \\ 0 &
0 & 0 \end{pmatrix},
\begin{pmatrix} 0 & 0 & 1 \\ 0 & 1 & 0 \\ 1 & 0 & 1
\end{pmatrix}) =
\begin{pmatrix} 1 & 1 & 0 \\ 0 & 0 & 1 \\ 0 & 0 & 0 \end{pmatrix}
$$
\caption{Basic matrix operations of $\ML$.
The matrix multiplication example is taken
from Axler's book \cite{ladoneright}.}
\label{figops}
\end{figure}

\begin{example}
The operations are illustrated in Figure~\ref{figops}.
In the pointwise application example,
we use the function $\dotmin$ defined by $x \dotmin y
= x - y$ if $x$ and $y$ are both real numbers and $x \geq y$, and
$x \dotmin y = 0$ otherwise.
\end{example}

\subsection{Formal semantics}

The formal semantics of expressions is defined in a
straightforward manner, as shown in Figure~\ref{fig:bigstep}.  An
\emph{instance} $I$ is a function, defined on a nonempty finite set
$\var(I)$ of matrix variables, that assigns a matrix to each
element of $\var(I)$.  Figure~\ref{fig:bigstep} provides the
rules that allow to derive that an expression $e$, on an
instance $I$, successfully evaluates to a matrix $A$.  We denote
this success by $\bigstep IeA$.  The reason why an evaluation may
not succeed can be found in the rules that have a condition
attached to them. The rule for variables fails when an instance
simply does not provide a value for some input variable.  The
rules for $\diag$, $\Apply$, and matrix
multiplication have conditions on the dimensions of matrices,
that need to be satisfied for the operations to be well-defined.

\begin{example}[Scalars]
\label{exscalar}
The expression from
Example~\ref{exex}, regardless of the matrix assigned to $M$, 
evaluates to the $1 \times 1$ matrix whose single entry equals
$c$.  We introduce the shorthand $c$ for this constant
expression.  Obviously, in practice, scalars would be
built in the language and would not be computed in such a roundabout manner.
In this paper, however, we are interested in expressiveness, so
we start from a minimal language and
then see what is already expressible in this language.
\end{example}

\begin{figure}
  \begin{mathpar}
    \infer{\M \in \var(I)}{\bigstep{I}{\M}{I(\M)}} \and
    \infer{\bigstep{I}{e_1}{A} \\ \bigstep{I[\M:=A]}{e_2}{B}}
    {\bigstep{I}{(\LetIn{\M}{e_1}{e_2})}{B}} \and
    \infer{\bigstep{I}{e}{A}}{\bigstep{I}{\transp{e}}{\transp{A}}} \and
    \infer{\bigstep{I}{e}{A}}{\bigstep{I}{\rowdom(e)}{\rowdom(A)}} \and
    \infer{\bigstep{I}{e}{A} \\ \text{$A$ is a column vector}}
    {\bigstep{I}{\diag(e)}{\diag(A)}} \and
    \infer{
      \bigstep{I}{e_1}{A} \\ \bigstep{I}{e_2}{B} \\
    \text{number of columns of $A$ equals the number of rows of $B$}
      }
    {\bigstep{I}{e_1 \cdot e_2}{A \cdot B}} \and
    \infer{\forall k = 1, \ldots, n : (\bigstep{I}{e_k}{A_k}) \\
      \text{all $A_k$ have the same dimensions}}
    {\bigstep{I}{\Apply[f](e_1, \ldots, e_n)}{\Apply[f](A_1, \ldots, A_n)}}
  \end{mathpar}
  \caption{Big-step operational semantics of $\ML$.
The notation $I[\M := A]$ denotes the instance that is equal to $I$,
except that $\M$ is mapped to the matrix $A$.}
  \label{fig:bigstep}
\end{figure}

\begin{example}[Scalar multiplication]
\label{exscalmul}
Let $A$ be any matrix and let $C$ be a $1 \times 1$ matrix; let $c$ be
the value of $C$'s single entry.  Viewing $C$ as a scalar, we define
the operation $C\odot A$ as multiplying every entry of $A$ by
$c$.  We can express $C \odot A$ as
$$ \LetIn{M}{\rowdom(A) \cdot C \cdot \transp{\rowdom(\transp{A})}}
{\Apply[\times](M,A)}. $$
If $A$ is an $m \times n$ matrix, we compute in variable $M$ the $m
\times n$ matrix where every entry equals $c$.  Then pointwise
multiplication is used to do the scalar multiplication.
\end{example}

\begin{example}[Google matrix] \label{ex:google-matrix}
Let $A$ be the adjacency matrix of a directed graph (modeling the Web graph)
on $n$ nodes numbered $1,\dots,n$.
Let $0 < d < 1$ be a fixed ``damping factor''.  Let $k_i$ denote the
outdegree of node $i$.  For simplicity, we assume $k_i$ is
nonzero for every $i$.  Then the Google matrix
\cite{pagerank,bonato_webgraph}
of $A$ is the $n \times n$ matrix $G$ defined by $$ G_{i,j} =
d \frac{A_{ij}}{k_i} + \frac{1-d}{n}. $$  The calculation
of $G$ from $A$ can be expressed in $\ML$ as follows:
\begin{tabbing}
\sffamily let $J = \one(A) \cdot \one(A)^*$ in \\
\sffamily let $K = A \cdot J$ in \\
\sffamily let $B = \Apply[/](A,K)$ in \\
\sffamily let $N = \one(A)^* \cdot \one(A)$ in \\
$\Apply[+]( d \odot B, (1-d) \odot \Apply[1/x](N) \odot J)$
\end{tabbing}
In variable $J$ we compute the $n \times n$ matrix where every
entry equals one.  In $K$ we compute the $n \times n$ matrix
where all entries in the $i$th row equal $k_i$.  In $N$ we
compute the $1 \times 1$ matrix containing the value $n$.  The
pointwise functions applied are addition, division, and reciprocal.
We use the shorthand for constants ($d$ and $1-d$) from
Example~\ref{exscalar}, and the shorthand $\odot$ for scalar
multiplication from Example~\ref{exscalmul}.
\end{example}

\begin{example}[Minimum of a vector] \label{exmin}
Let $v=(v_1,\dots,v_n)^*$
be a column vector of real numbers; we would like to
extract the minimum from $v$.  This can be done as follows:
\begin{tabbing}
\sffamily let $V = v \cdot \one(v)^*$ in \\
\sffamily let $C = \Apply[\leq](V,V^*) \cdot \one(v)$ in \\
\sffamily let $N = \one(v)^* \cdot \one(v)$ in \\
\sffamily let $S = \Apply[=](C,\one(v) \cdot N)$ in \\
\sffamily let $M = \Apply[1/x](S^* \cdot \one(v))$ in \\
$M \cdot v^* \cdot S$
\end{tabbing}
The pointwise functions applied are $\leq$, which returns 1
on $(x,y)$ if $x \leq y$ and $0$ otherwise;
$=$, defined analogously; and the reciprocal function.
In variable $V$ we compute a square matrix holding $n$ copies of
$v$.  Then in variable $C$
we compute the $n \times 1$ column vector where $C_i$
counts the number of $v_j$ such that $v_i \leq v_j$. If $C_i=n$
then $v_i$ equals the minimum.  Variable $N$ computes the scalar
$n$ and column vector $S$ is a selector where $S_i=1$ if $v_i$
equals the minimum, and $S_i=0$ otherwise.  Since the minimum may
appear multiple times in $v$, we compute in $M$ the inverse of
the multiplicity.  Finally we sum the different
occurrences of the minimum in $v$ and divide by the multiplicity.
\end{example}

\subsection{Types and schemas} \label{sectypes}

We have already remarked that, due to conditions on the
dimensions of matrices, $\ML$ expressions are not well-defined on
all instances.  For example, if $I$ is an instance where $I(M)$
is a $3 \times 4$ matrix and $I(N)$ is a $2 \times 4$ matrix,
then the expression $M \cdot N$ is not defined on $I$.  The
expression $M \cdot N^*$, however, is well-defined on $I$.  We
now introduce a notion of schema, which assigns types to matrix
names, so that expressions can be type-checked against schemas.

Our types need to be able to guarantee equalities between numbers
of rows or numbers of columns, so that $\Apply$ and matrix
multiplication can be typechecked.  Our types also need to be
able to recognize vectors, so that $\diag$ can be typechecked.

Formally, we assume a sufficient supply of \emph{size symbols},
which we will denote by the letters $\alpha$, $\beta$, $\gamma$.
A size symbol represents the number of rows or columns of a
matrix.  Together with an explicit 1, we can indicate
arbitrary matrices as $\alpha \times \beta$, square matrices as
$\alpha \times \alpha$, column vectors as $\alpha \times 1$, row
vectors as $1 \times \alpha$, and scalars as $1 \times 1$.
Formally, a \emph{size term} is either a size symbol or an
explicit 1.  A \emph{type} is then an expression of the form $s_1
\times s_2$ where $s_1$ and $s_2$ are size terms.  Finally, a
\emph{schema} $\scm$ is a function, defined on a nonempty finite
set $\var(\scm)$ of matrix variables, that assigns a type to each
element of $\var(\scm)$.

The typechecking of expressions is now shown in
Figure~\ref{fig:matlang-type-rules}.  The figure provides the rules
that allow to infer an output type $\tau$ for an expression $e$
over a schema $\scm$.  To indicate that a type can be
successfully inferred, we use the notation
$\typed{\scm}{e}{\tau}$.  When we cannot infer a type, we say $e$
is not well-typed over $\scm$.  For example, when
$\scm(M)=\alpha \times \beta$ and $\scm(N) = \gamma \times
\beta$, then the expression $M \cdot N$ is not well-typed over
$\scm$.  The expression $M \cdot N^*$, however, is well-typed
with output type $\alpha \times \gamma$.

To establish the soundness of the type system,
we need a notion of conformance of an instance to a schema.

Formally, a \emph{size assignment} $\sigma$ is a function
from size symbols to positive natural numbers.  We extend
$\sigma$ to any size term by setting $\sigma(1) = 1$.  Now, let
$\scm$ be a schema and $I$ an instance with $\var(I) =
\var(\scm)$.  We say that $I$ is an instance of $\scm$
if there is a size assignment $\sigma$ such
that for all $\M \in \var(\scm)$, if $\scm(\M) = s_1 \times s_2$,
then $I(\M)$ is a $\sigma(s_1) \times \sigma(s_2)$ matrix.  In
that case we also say that $I$
\emph{conforms} to $\scm$ by the size assignment $\sigma$.

We now obtain the following obvious but desirable property.
\begin{proposition}[Safety]
  \label{prop:matlang-safety}
If $\scm \vdash e : s_1 \times s_2$, then for every instance $I$
conforming to $\scm$, by size assignment $\sigma$, the
matrix $e(I)$ is well-defined and has dimensions
$\sigma(s_1) \times \sigma(s_2)$.
\end{proposition}

\begin{figure}
  \begin{mathpar}
    \infer{\M \in \var(\scm)}
    {\typed{\scm}{\M}{\scm(\M)}} \and
    \infer{\typed{\scm}{e_1}{\tau_1} \\ \typed{\scm[\M := \tau_1]}{e_2}{\tau_2}}
    {\typed{\scm}{\LetIn{\M}{e_1}{e_2}}{\tau_2}} \and
    \infer{\typed{\scm}{e}{s_1 \times s_2}}
    {\typed{\scm}{\transp{e}}{s_2 \times s_1}} \and
    \infer{\typed{\scm}{e}{s_1 \times s_2}}
    {\typed{\scm}{\rowdom(e)}{s_1 \times 1}} \and
    \infer{\typed{\scm}{e}{s \times 1}}
    {\typed{\scm}{\diag(e)}{s \times s}} \and
    \infer{\typed{\scm}{e_1}{s_1 \times s_2} \\ \typed{\scm}{e_2}{s_2 \times s_3}}
    {\typed{\scm}{e_1 \cdot e_2}{s_1 \times s_3}} \and
    \infer{n > 0 \\ f : \C^n\to \C \\ \forall k = 1, \ldots, n : (\typed{\scm}{e_k}{\tau})}
    {\typed{\scm}{\Apply[f](e_1, \ldots, e_n)}{\tau}}
  \end{mathpar}
  \caption{Typechecking $\ML$.
The notation $\scm[\M := \tau]$ denotes the schema that is equal to $\scm$,
except that $\M$ is mapped to the type $\tau$.}
  \label{fig:matlang-type-rules}
\end{figure}

\section{Expressive power of $\ML$} \label{secsum}

It is natural to represent an $m \times n$ matrix $A$
by a ternary relation $$
\Rel_2(A) := \{(i,j,A_{i,j}) \mid i \in \{1,\dots,m\}, \ j \in
\{1,\dots,n\}\}. $$  In the special case where $A$ is an
$m \times 1$ matrix (column
vector), $A$ can also be represented by a binary relation
$\Rel_1(A) :=
\{(i,A_{i,1}) \mid i \in \{1,\dots,m\}\}$.  Similarly, a $1
\times n$ matrix (row vector) $A$ can be represented by $\Rel_1(A)
:= \{(j,A_{1,j}) \mid j \in \{1,\dots,n\}\}$.  Finally, a $1
\times 1$ matrix (scalar) $A$ can be represented by the unary
singleton relation $\Rel_0(A) := \{(A_{1,1})\}$.

Note that in $\ML$, we perform calculations on matrix entries,
but not on row or column indices.  This fits well to the
relational model with aggregates as formalized by Libkin
\cite{libkin_sql}.  In this model, the columns of relations are
typed as ``base'', indicated by $\Tb$, or ``numerical'',
indicated by $\Tn$.  In the relational representations of
matrices presented above, the last column is of type $\Tn$ and
the other columns (if any) are of type $\Tb$.  In particular, in
our setting, numerical columns hold complex numbers.

Given this representation of matrices by relations, $\ML$
can be simulated in the relational
algebra with aggregates.  Actually, the only aggregate operation
we need is summation.  We will not reproduce the formal definition of the
relational algebra with summation \cite{libkin_sql},
but note the following salient points:
\begin{itemize}
\item
Expressions are built up from relation names using the classical operations
union, set difference, cartesian product
($\times$), selection ($\sigma$), and projection ($\pi$), plus
two new operations: \emph{function application} and
\emph{summation}.
\item
For selection, we only use
equality and nonequality comparisons on base columns.  No
selection on numerical columns will be needed in our setting.
\item
For any function $f:\C^n \to \C$, the
operation $\Apply[f;i_1,\dots,i_n]$ can be applied to any
relation $r$ having columns $i_1$, \dots, $i_n$, which
must be numerical.  The result
is the relation $\{(t,f(t({i_1}),\dots,t({i_n}))) \mid t \in
r\}$, adding a numerical column to $r$.
We allow $n=0$, in which case $f$ is a constant.
\item
\newcommand{\group}{\mathsf{group}}
The operation $\Sum[i;i_1,\dots,i_n]$ can be applied to any
relation $r$ having columns $i$, $i_1$, \dots, $i_n$, where
column $i$ must be numerical.  In our setting we only need the operation in
cases where columns $i_1$, \dots, $i_n$ are base columns.
The result of the operation is the relation $$
\{(t(i_1),\dots,t(i_n),\sum_{t' \in \group[i_1,\dots,i_n](r,t)}
t'(i)) \mid t \in r\}, $$ where
$$ \group[i_1,\dots,i_n](r,t) = \{t' \in r \mid t'(i_1)=t(i_1)
\land \cdots \land t'(i_n)=t(i_n)\}. $$
Again, $n$ can be zero, in which case the result is a singleton.
\end{itemize}

\subsection{From $\ML$ to relational algebra with summation}

To state the translation formally,
we assume a supply of \emph{relation variables}, which, for
convenience, we can take to be the same as the matrix variables.
A \emph{relation type} is a tuple of $\Tb$'s and $\Tn$'s.
A \emph{relational schema} $\scm$ is a function, defined on a
nonempty finite set
$\var(\scm)$ of relation variables, that assigns a relation type
to each element of $\var(\scm)$.

One can define well-typedness for expressions in the relation
algebra with summation, and define the output type.  We omit this
definition here, as it follows a well-known methodology
\cite{crash} and is analogous to what we have already done for
$\ML$ in Section~\ref{sectypes}.

To define relational instances, we assume a countably infinite universe
$\dom$ of abstract atomic data elements.  It is convenient to
assume that the natural numbers are contained in $\dom$.  We
stress that this assumption is not essential but simplifies the
presentation.  Alternatively, we would have to work with explicit
embeddings from the natural numbers into $\dom$.

Let $\tau$ be a relation type. A \emph{tuple
of type $\tau$} is a tuple $(t(1),\dots,t(n))$ of the same arity
as $\tau$, such that $t(i) \in \dom$ when $\tau(i) = \Tb$, and
$t(i)$ is a complex number when $\tau(i) = \Tn$.
A \emph{relation of type
$\tau$} is a finite set of tuples of type $\tau$.  
An \emph{instance} of a relational schema $\scm$ is a
function $I$ defined on $\var(\scm)$ so that $I(R)$ is a relation
of type $\scm(R)$ for every $R \in \var(\scm)$.

We must connect the matrix data model to the relational data
model.  Let $\tau = s_1\times s_2$ be a matrix type.  Let us call $\tau$ a
\emph{general type} if $s_1$ and $s_2$ are both size symbols; a
\emph{vector type} if $s_1$ is a size symbol and $s_2$ is 1, or
vice versa; and the \emph{scalar type} if $\tau$ is $1\times 1$.
To every matrix type $\tau$ we associate a relation type
$$ \Rel(\tau) := \begin{cases}
(\Tb,\Tb,\Tn) & \text{if $\tau$ is general;} \\
(\Tb,\Tn) & \text{if $\tau$ is a vector type;} \\
(\Tn) & \text{if $\tau$ is scalar.} \end{cases} $$
Then to every matrix schema $\scm$ we associate the relational
schema $\Rel(\scm)$ where $\Rel(\scm)(M) = \Rel(\scm(M))$ for
every $M \in \var(\scm)$.  For each instance $I$ of
$\scm$, we define the instance $\Rel(I)$ over
$\Rel(\scm)$ by $$ \Rel(I)(M) = \begin{cases} 
\Rel_2(I(M)) & \text{if $\scm(M)$ is a general type;} \\
\Rel_1(I(M)) & \text{if $\scm(M)$ is a vector type;} \\
\Rel_0(I(M)) & \text{if $\scm(M)$ is the scalar type.}
\end{cases} $$  Here we use the relational representations
$\Rel_2$, $\Rel_1$ and $\Rel_0$ of matrices introduced in the
beginning of Section~\ref{secsum}.

\begin{theorem} \label{sumtheorem}
Let $\scm$ be a matrix schema, and let $e$ a $\ML$ expression
that is well-typed over $\scm$ with output type $\tau$.
Let $\ell=2$, $1$,
or $0$, depending on whether $\tau$ is general, a vector type, or
scalar, respectively.
\begin{enumerate}
\item
There
exists an expression $\Rel(e)$ in the relational algebra with
summation, that is well-typed over $\Rel(\scm)$ with output type
$\Rel(\tau)$, such that for every instance $I$ of $\scm$, we have
$\Rel_\ell(e(I)) = \Rel(e)(\Rel(I))$.
\item
The expression $\Rel(e)$ uses
neither set difference, nor selection conditions on numerical
columns.
\item
The only functions used in $\Rel(e)$ are those used in
pointwise applications in $e$; complex conjugation;
multiplication of two numbers;
and the constant functions $0$ and $1$.
\end{enumerate}
\end{theorem}
\begin{proof}
We only give a few representative examples.
\begin{itemize}
\item
If $M$ is of type $\alpha \times \beta$ then $\Rel(M^*)$ is
$\Apply[\overline z;3] \, \pi_{2,1,3}(M)$, where $\overline z$ is
the complex conjugate. If $M$ is of type $\alpha \times 1$, however,
$\Rel(M^*)$ is $\Apply[\overline z;2](M)$.
\item
If $M$ is of type $1 \times \alpha$ then $\Rel(\one(M))$ is
$\pi_3(\Apply[1;2](M))$.  Here, $1$ stands for the constant $1$
function.
\item
If $M$ is of type $\alpha \times 1$ then $\Rel(\diag(M))$ is 
$$\sigma_{\$1=\$2}(\pi_1(M) \times M) \cup
\Apply[0;\,] \, \sigma_{\$1\neq \$2}(\pi_1(M) \times \pi_1(M)).$$
\item
If $M$ is of type $\alpha \times \beta$ and $N$ is of type $\beta
\times \gamma$, then $\Rel(M \cdot N)$ is $$ \Sum[7;1,5] \,
\Apply[\times;3,6] \, \sigma_{\$2=\$4}(M \times N).$$  If,
however, $M$ is of type $\alpha \times 1$ and $N$ is of type $1
\times 1$, then $\Rel(M \cdot N)$ is $$\pi_{1,4} \, \Apply[\times;2,3](M
\times N).$$  We use pointwise multiplication.
\item
If $M$ and $N$ are of type $1 \times \beta$ then
$\Rel(\Apply[f](M,N))$ is $\pi_{1,5} \, \Apply[f;2,4] \,
\sigma_{\$1=\$3}(M \times N)$.
\end{itemize}
We may ignore the let-construct as it does not add expressive
power.
\end{proof}

\begin{remark}
The different treatment of general types, vector types, and
scalar types is necessary because in our version of the
relational algebra, selections can only compare base columns for equality;
in particular we can not select for the value 1.
\end{remark}

\begin{remark}
We can sharpen the above theorem a bit if we work in the
relational calculus with aggregates.  Every $\ML$ expression can
already be expressed by a formula in the relational calculus with
summation that uses only three distinct base variables (variables
ranging over values in base columns).  The details are given in
the Appendix.
\end{remark}

\subsection{Expressing graph queries}

So far we have looked at expressing matrix queries in terms of
relational queries.  It is also natural to express relational
queries as matrix queries.  This works best for binary relations,
or graphs, which we can represent by their adjacency matrices.  

Formally, define a \emph{graph schema} to be a relational schema
where every relation variable is assigned the type $(\Tb,\Tb)$ of
arity two.  We define a \emph{graph instance} as an instance $I$
of a graph schema, where the active domain of $I$ equals
$\{1,\dots,n\}$ for some positive natural number $n$.  The
assumption that the active domain always equals an initial
segment of the natural numbers is convenient for forming the
bridge to matrices.  This assumption, however, is not essential
for our results to hold.  Indeed, the logics we consider do not
have any built-in predicates on base variables, besides equality.
Hence, they view the active domain elements as abstract data
values.

To every graph schema $\scm$ we associate a matrix schema
$\Mat(\scm)$, where $\Mat(\scm)(R) = \alpha \times \alpha$ for
every $R \in \var(\scm)$, for a fixed size symbol $\alpha$.  So,
all matrices are square matrices of the same dimension.  Let $I$
be a graph instance of $\scm$, with active domain
$\{1,\dots,n\}$.  We will denote the $n \times n$ adjacency
matrix of a binary relation $r$ over $\{1,\dots,n\}$ by
$\Adj_I(r)$.  Now any such instance $I$ is represented by the
matrix instance $\Mat(I)$ over $\Mat(\scm)$, where $\Mat(I)(R) =
\Adj_I(I(R))$ for every $R \in \var(\scm)$.

A \emph{graph query} over a graph schema $\scm$ is a function
that maps each graph instance $I$ of $\scm$ to a binary relation
on the active domain of $I$.  We say that a $\ML$ expression $e$
\emph{expresses} the graph query $q$ if $e$ is well-typed over
$\Mat(\scm)$ with output type $\alpha \times \alpha$, and for
every graph instance $I$ of $\scm$, we have $\Adj_I(q(I)) =
e(\Mat(I))$.

We can now give a partial converse to Theorem~\ref{sumtheorem}.
We assume active-domain semantics for first-order logic
\cite{ahv_book}.  Please note that the following result deals
only with pure first-order logic, without aggregates or numerical
columns.  The proof, while instructive, has been relegated to the
Appendix.

\begin{theorem} \label{fo3}
Every graph query expressible in $\rm FO^3$ (first-order
logic with equality, using at most three distinct
variables) is expressible in $\ML$.
The only functions needed in pointwise
applications are boolean functions on $\{0,1\}$, and testing if a
number if positive.
\end{theorem}

We can complement the above theorem by showing that the
quintessential first-order query requiring four variables is not
expressible.  The proof is given in the Appendix.

\begin{proposition} \label{4clique}
The graph query over a single binary relation $R$ that maps $I$
to $I(R)$ if $I(R)$ contains a four-clique, and to the empty
relation otherwise, is not expressible in $\ML$.
\end{proposition}

\section{Matrix inversion} \label{secinv}

Matrix inversion (solving nonsingular
systems of linear equations) is an ubiquitous
operation in data analysis.  We can extend $\ML$ with matrix
inversion as follows.  Let $\scm$ be a schema and $e$ be an
expression that is well-typed over $\scm$, with output type of
the form $\alpha \times \alpha$.  Then the expression $e^{-1}$ is
also well-typed over $\scm$, with the same output type $\alpha
\times \alpha$.  The semantics is defined as follows.  For an
instance $I$, if $e(I)$ is an invertible matrix, then $e^{-1}(I)$
is defined to be the inverse of $e(I)$; otherwise, it is defined
to be the zero square matrix of the same dimensions as $e(I)$.
The extension of $\ML$ with inversion is denoted by $\ML + \inv$.

\begin{example}[PageRank]
Recall Example~\ref{ex:google-matrix} where we computed the
Google matrix of $A$.  In the process we already showed how to
compute the $n \times n$ matrix $B$ defined by $B_{i,j} =
A_{i,j}/k_i$, and the scalar $N$ holding the value $n$.
So, in the
following expression, we assume we already have $B$ and $N$.
Let $I$ be the $n \times n$ identity matrix, and
let $\one$ denote the $n \times 1$ column vector consisting of all
ones.  The PageRank vector $v$ of $A$ can be computed as follows
\cite{delcorso}:
$$ v = \frac{1-d}n(I - dB)^{-1} \one. $$  This calculation is
readily expressed in $\ML + \inv$ as
$$ (1-d) \odot \Apply[1/x](N) \odot \Apply[-](\diag(\one(A)),d
\odot B)^{-1} \cdot \one(A). $$
\end{example}

\begin{example}[Transitive closure] \label{tc}
We next show that the reflexive-transitive closure of a binary
relation is expressible in $\ML +
\inv$.  Let $A$ be the adjacency matrix of a binary relation $r$ on
$\{1,\dots,n\}$.   Let $I$ be the $n \times n$ identity matrix,
expressible as $\diag(\one(A))$.
From earlier examples we know how to compute the scalar $1 \times
1$ matrix $N$ holding the value $n$.
The matrix $B = \frac 1{n+1} A$ has 1-norm
strictly less than 1, so $S = \sum_{k=0}^\infty B^k$ converges,
and is equal to $(I - B)^{-1}$
\cite[Lemma~2.3.3]{golub}.  Now $(i,j)$ belongs to the
reflexive-transitive closure of $r$ if and only if $S_{i,j}$ is
nonzero.  Thus, we can express the
reflexive-transitive closure of $r$ as
$$ \Apply[\neq 0] \bigl (
\Apply[-](\diag(\one(A)) ,
\Apply[1/(x+1)](N) \odot A)^{-1} \bigr ), $$
where $x \neq 0$ is $1$ if $x \neq 0$ and $0$ otherwise.
We can obtain the transitive closure by multiplying the above
expression with $A$.
\qed
\end{example}

By Theorem~\ref{sumtheorem}, any graph query expressible in $\ML$ is
expressible in the relational algebra with aggregates.  It is
known \cite{hlnw_aggregate,libkin_sql} that such queries are
local.  The transitive-closure query from Example~\ref{tc},
however, is not local.  We thus conclude:

\begin{theorem}
$\ML + \inv$ is strictly more powerful than $\ML$ in expressing
graph queries.
\end{theorem}

Once we have the transitive closure, we can do many
other things such as checking bipartiteness of undirected graphs,
checking connectivity, checking cyclicity.  $\ML$ is expressive
enough to reduce these queries to the transitive-closure query,
as shown in the following example for bipartiteness.  The
same approach via $\rm FO^3$  can be used for connectedness or cyclicity.

\begin{example}[Bipartiteness]
To check bipartiteness of an undirected graph, given as a
symmetric binary relation $R$ without self-loops, we
first compute the transitive closure $T$ of the composition of
$R$ with itself.  Then the
$\rm FO^3$ condition $\neg \exists x \exists y (R(x,y)
\land T(y,x))$ expresses that $R$ is bipartite (no odd cycles).
The result now follows from Theorem~\ref{fo3}.
\end{example}

\begin{example}[Number of connected components]
Using transitive closure we can also easily compute the number of
connected components of a binary relation $R$ on $\{1,\dots,n\}$, given as an
adjacency matrix.  We start from the union of $R$ and its
converse. This union, denoted by $S$,
is expressible by Theorem~\ref{fo3}.  We then
compute the reflexive-transitive
closure $C$ of $S$.  Now the number of connected
components of $R$ equals $\sum_{i=1}^n 1/k_i$, where $k_i$ is the
degree of node $i$ in $C$.  This sum is simply expressible as $\one(C)^*
\cdot \Apply[1/x](C \cdot \one(C))$.
\end{example}

\section{Eigenvalues} \label{seceigen}

Another workhorse in data analysis is diagonalizing a matrix,
i.e., finding a basis of eigenvectors.  Formally, we define the
operation $\eigen$ as follows.  Let $A$ be an $n \times n$
matrix.  Recall that $A$ is called diagonalizable if there exists
a basis of $\C^n$ consisting of eigenvectors of $A$.  In that
case, there also exists such a basis where eigenvectors
corresponding to a same eigenvalue are orthogonal. Accordingly,
we define $\eigen(A)$ to return an $n \times n$ matrix, the
columns of which form a basis of $\C^n$ consisting of
eigenvectors of $A$, where eigenvectors corresponding to a same
eigenvalue are orthogonal.  If $A$ is not diagonalizable,
we define $\eigen(A)$ to be the $n \times n$ zero matrix.

Note that $\eigen$ is nondeterministic; in principle there are
infinitely many possible results.  This models the situation in
practice where numerical packages such as R or \textsf{MATLAB}
return approximations to the eigenvalues and a set of
corresponding eigenvectors, but the latter are not unique.
Hence, some care must be taken in extending $\ML$ with the
$\eigen$ operator.  Syntactically, as for inversion, whenever $e$
is a well-typed expression with a square output type, we now also
allow the expression $\eigen(e)$, with the same output type.
Semantically, however, the rules of Figure~\ref{fig:bigstep} must
be adapted so that they do not infer statements of the form
$e(I)=B$, but rather of the form $B \in e(I)$, i.e., $B$ is a
possible result of $e(I)$.  The let-construct now becomes
crucial; it allows us to assign a possible result of $\eigen$ to
a new variable, and work with that intermediate result
consistently.

In this and the next section, we assume notions from linear algebra. An
excellent introduction to the subject has been given by Axler
\cite{ladoneright}.

\begin{remark}[Eigenvalues]
We can easily recover the eigenvalues from the eigenvectors,
using inversion.  Indeed, if $A$ is diagonalizable and $B \in
\eigen(A)$, then $\Lambda=B^{-1} A B$ is a diagonal matrix
with all eigenvalues of $A$ on the diagonal, so that the
$i$th eigenvector in $B$ corresponds to the eigenvalue in the $i$th
column of $\Lambda$.  This is the well-known eigendecomposition.
However, the same can also be accomplished without using inversion.  
Indeed, suppose $B = (v_1,\dots,v_n)$, and let $\lambda_i$ be the
eigenvalue to which $v_i$ corresponds.  Then $A B =
(\lambda_1v_1,\dots,\lambda_nv_n)$.  Each eigenvector is nonzero,
so we can divide away the entries from $B$ in $A B$ (setting
division by zero to zero).  We thus
obtain a matrix where the $i$th column consists of zeros or
$\lambda_i$, with at least one occurrence of $\lambda_i$.  By
counting multiplicities, dividing them out, and finally summing,
we obtain $\lambda_1$, \dots, $\lambda_n$ in a column vector.
We can apply a final $\diag$ to get it back into diagonal form.
The $\ML$ expression for doing all this uses similar tricks as those shown
in Examples \ref{ex:google-matrix} and \ref{exmin}.
\qed
\end{remark}

The above remark suggests a shorthand in $\ML + \eigen$ where we
return both $B$ and $\Lambda$ together: $$ \mathsf{let} \
(B,\Lambda) = \eigen(A) \ \mathsf{in} \ \dots $$ This models how
the $\eigen$ operation works in the languages R and
\textsf{MATLAB}.  We agree that $\Lambda$, like $B$, is the zero
matrix if $A$ is not diagonalizable.

\begin{example}[Rank of a matrix] \label{exrank}
Since the rank of a diagonalizable matrix equals the number of
nonzero entries in its diagonal form, we can express the rank of
a diagonalizable matrix $A$ as follows:
$$
\mathsf{let} \ (B,\Lambda) =\eigen(A) \ \mathsf{in} \
\one(A)^* \cdot \Apply[\neq 0](\Lambda) \cdot \one(A). $$
\end{example}

\begin{example}[Graph partitioning]
A well-known heuristic for partitioning an undirected graph
without self-loops
is based on an eigenvector corresponding to the second-smallest
eigenvalue of the Laplacian matrix \cite{ullman_mining}.  
The Laplacian $L$ can be derived from the adjacency matrix $A$
as \textsf{let $D = \diag(A \cdot \one(A))$ in $\Apply[-](D,A)$}.
(Here $D$ is the degree matrix.)  Now let $(B,\Lambda) \in
\eigen(L)$.
In an analogous way to Example~\ref{exmin}, we can compute a
matrix $E$, obtained from $\Lambda$ by replacing the occurrences
of the second-smallest eigenvalue by 1 and all other entries
by 0.  Then the eigenvectors corresponding to this eigenvalue can
be isolated from $B$ (and the other eigenvectors zeroed out) by
multiplying $B \cdot E$.
\qed
\end{example}

It turns out that $\ML + \inv$ is subsumed by $\ML + \eigen$.
The proof is in the Appendix.

\begin{theorem} \label{inv2eigen}
Matrix inversion is expressible in $\ML + \eigen$.
\end{theorem}

A natural question to ask is if $\ML$ with $\eigen$ is strictly more
expressive than $\ML$ with $\inv$.  In a noninteresting sense,
the answer is affirmative.  Indeed, when evaluating a $\ML + \inv$ expression
on an instance where all matrix entries are rational numbers, the
result matrix is also rational.  In contrast, the eigenvalues of
a rational matrix may be complex numbers.  The more interesting
question, however, is: \emph{Are there graph queries
expressible deterministically in $\ML + \eigen$,
but not in $\ML + \inv$?}  This is an
interesting question for further research.  The answer may depend
on the functions that can be used in pointwise applications.

\begin{remark}[Determinacy]
The stipulation \emph{deterministically} in the above open question is
important.  Ideally, we use the nondeterministic $\eigen$
operation only as an intermediate construct.  It is an aid to
achieve a powerful computation, but the
final expression should have only a single possible output on
every input.
The expression of Example~\ref{exrank} is deterministic in this
sense, as is the expression for inversion described in the proof
of Theorem~\ref{inv2eigen}.
\end{remark}

\section{The evaluation problem}
\label{seceval}

\newcommand{\FV}{\mathrm{FV}}

The evaluation problem asks, given an input instance $I$ and an
expression $e$, to compute the result $e(I)$.  There are some
issues with this naive formulation, however.  Indeed, in our
theory we have been working with arbitrary complex numbers.  How
do we even represent the input?  For practical applications, it
is usually sufficient to support matrices with rational numbers
only.  For $\ML + \inv$, this approach works: when the input is
rational, the output is rational too, and can be computed in
polynomial time.  For the basic matrix operations this is clear,
and for matrix inversion we can use the well known method of
Gaussian elimination.

When adding the $\eigen$ operation, however, the output may
become irrational.  Much worse, the eigenvalues of an adjacency
matrix (even of a tree) need not even be definable in radicals
\cite{godsil_radicals}.  Practical systems, of course, apply
techniques from numerical mathematics to compute rational
approximations.  But it is still theoretically interesting to
consider the exact evaluation problem.

Our approach is to represent the output symbolically, following
the idea of constraint query languages \cite{kkr_cql,cdbbook}.
Specifically,
we can define the input-output relation of an expression, for given
dimensions of the input matrices, by an existential first-order
logic formula over the reals.  Such formulas are
built from real variables, integer constants, addition,
multiplication, equality, inequality ($<$), disjunction,
conjunction, and existential quantification.

\begin{example} \label{exeigenformula}
Consider the expression $\eigen(M)$ over the schema
consisting of a single matrix variable $M$.  Any instance $I$
where $I(M)$ is an $n \times n$ matrix $A$ can be represented by a
tuple of $2 \times n \times n$ real numbers.
Indeed, let $a_{i,j} = \Re A_{i,j}$ (the real part of a complex
number), and let $b_{i,j} = \Im A_{i,j}$ (the imaginary part).
Then $I(M)$ can be represented
by the tuple $(a_{1,1},b_{1,1},a_{1,2},b_{1,2},\dots,a_{n,n},b_{n,n})$.
Similarly, any $B \in \eigen(A)$ can be represented by a similar
tuple.  We introduce the
variables $x_{M,i,j,\Re}$, $x_{M,i,j,\Im}$, $y_{i,j,\Re}$, and
$y_{i,j,\Im}$, for $i,j \in \{1,\dots,n\}$, where the $x$-variables
describe an arbitrary input matrix and the $y$-variables describe
an arbitrary
possible output matrix.  Denoting the input matrix by $[\bar x]$
and the output matrix by $[\bar y]$, we can now write an
existential formula expressing that
$[\bar y]$ is a possible result of $\eigen$ applied to $[\bar x]$:
\begin{itemize}
\item
To express that $[\bar y]$ is a basis, we write that there exists
a nonzero matrix $[\bar z]$ such that $[\bar y] \cdot [\bar z]$
is the identity matrix.  It is straightforward to express this
condition by a formula.
\item
To express, for each column vector $v$ of $[\bar y]$, that $v$ is an
eigenvector of $[\bar x]$, we write that there exists $\lambda$
such that $[\bar x] \cdot v = \lambda[\bar x]$.
\item
The final and most difficult condition to express is that
distinct eigenvectors $v$ and $w$ that correspond to a same eigenvalue are
orthogonal.  We cannot write $\exists \lambda ([\bar x]\cdot v =
\lambda v \land [\bar x]\cdot w = \lambda w) \to v^* \cdot w =
0$, as this is not a proper existential formula.  (Note though that the
conjugate transpose of $v$ is readily expressed.)  Instead, we
avoid an explicit quantifier and rewrite the antecedent as
the conjunction, over all positions $i$, of $v_i
\neq 0 \neq w_i \to
([\bar x] \cdot v)_i / v_i = 
([\bar x] \cdot w)_i / w_i$.
\item
A final detail is that we should also be able to express that
$[\bar x]$ is not diagonalizable, for in that case we need to
define $[\bar y]$ to be the zero matrix.  Nondiagonalizability
is equivalent to the existence of
a Jordan form with at least one 1 on
the superdiagonal.  We can express this as follows. We postulate the
existence of an invertible matrix $[\bar z]$ such that the
product $[\bar z] \cdot [\bar x] \cdot [\bar z]^{-1}$
has all entries zero, except those on the diagonal and
the superdiagonal.  The entries on the superdiagonal can only by
0 or 1, with at least one 1.  Moreover, if an entry $i,j$ on
  the superdiagonal is nonzero, the entries $i,i$ and
  $j,j$ must be equal.
\qed
\end{itemize}
\end{example}

The approach taken in the above example leads to the following
general result.  The operations of $\ML$ are handled using
similar ideas as illustrated above for the $\eigen$ operation,
and are actually easier.  The let-construct, and the composition
of subexpressions into larger expression, are handled by
existential quantification.

\begin{theorem} \label{existsrtheorem}
An input-sized expression consists of a schema $\scm$, an expression $e$
in $\ML + \eigen$ that is well-typed over $\scm$ with output type
$t_1 \times t_2$, and a size assignment $\sigma$ defined on the
size symbols occurring in $\scm$.  There exists a polynomial-time
computable translation that maps any input-sized expression as above to
an existential first-order formula $\psi$ over the vocabulary of
the reals, expanded with symbols for the functions used in
pointwise applications in $e$, such that
\begin{enumerate}
\item
Formula $\psi$ has the following free variables:
\begin{itemize}
\item
For every $M \in \var(\scm)$, let $\scm(M)=s_1\times s_2$.  Then
$\psi$ has the free variables
$x_{M,i,j,\Re}$ and $x_{M,i,j,\Im}$,
for $i=1,\dots,\sigma(s_1)$
and $j=1,\dots,\sigma(s_2)$.
\item
In addition, $\psi$ has the free variables
$y_{i,j,\Re}$ and $y_{i,j,\Im}$, 
for $i=1,\dots,\sigma(t_1)$
and $j=1,\dots,\sigma(t_2)$.
\end{itemize}
The set of these free variables is denoted by
$\FV(\scm,e,\sigma)$.
\item
Any assignment $\rho$ of real numbers to these variables
specifies, through the
$x$-variables, an instance
$I$ conforming to $\scm$ by $\sigma$, and through the
$y$-variables, a $\sigma(t_1)\times \sigma(t_2)$ matrix $B$.
\item
Formula $\psi$ is true over the reals under such an assignment
$\rho$, if and only if $B \in e(I)$.
\end{enumerate}
\end{theorem}

The existential theory of the reals is decidable; actually, the
full first-order theory of the reals is decidable
\cite{arnon,basu_algorithms}.
But, specifically the class of problems that can be
reduced in polynomial time to the existential theory of the reals
forms a complexity class on its own, known as $\exists \mathbf R$
\cite{schaefer_existsR,schaefer_nash}.  The above theorem implies 
that the \emph{partial evaluation problem for $\ML + \eigen$}
belongs to this complexity class.  We define this problem as
follows.  The idea is that an arbitrary specification, expressed
as an existential formula $\chi$ over the reals, can be imposed on the
input-output relation of an input-sized expression.

\begin{definition}
The \emph{partial evaluation problem} is a decision problem that
takes as input:
\begin{itemize}
\item
an input-sized expression $(\scm,e,\sigma)$, where all functions
used in pointwise applications are explicitly defined using
existential formulas over the reals;
\item
an existential formula $\chi$ with free variables in
$\FV(\scm,e,\sigma)$ (see
Theorem~\ref{existsrtheorem}).
\end{itemize}
The problem asks if there
exists an instance $I$ conforming to $\scm$ by $\sigma$ and a matrix
$B \in e(I)$ such that $(I,B)$ satisfies $\chi$.
\end{definition}
For example, $\chi$ may completely specify the matrices in $I$ by
giving the values of the entries as rational numbers, and may
express that the output matrix has at least one nonzero entry.

An input $(\scm,e,\sigma,\chi)$ is a yes-instance 
to the partial evaluation problem precisely when the existential sentence
$\exists \FV(\scm,e,\sigma) (\psi \land \chi)$ is true in the
reals, where $\psi$ is the formula obtained by
Theorem~\ref{existsrtheorem}.  Hence we can conclude:

\begin{corollary} \label{upperbound}
The partial evaluation problem for $\ML + \eigen$ belongs to
$\exists \mathbf R$.
\end{corollary}

Since the full theory of the reals is decidable, our theorem
implies many other decidability results.  We give just two examples.

\begin{corollary}
The equivalence problem for input-sized expressions is decidable.
This problem takes as input two input-sized expressions
$(\scm,e_1,\sigma)$ and $(\scm,e_2,\sigma)$ (with the same $\scm$
and $\sigma$) and asks if for all instances $I$ conforming to
$\scm$ by $\sigma$, we have $B \in e_1(I) \; \Leftrightarrow \; B
\in e_2(I)$.
\end{corollary}
Note that the equivalence problem for $\ML$ expressions on
arbitrary instances (size not fixed) is undecidable by
Theorem~\ref{fo3}, since equivalence of $\rm FO^3$ formulas over
binary relational vocabularies is undecidable \cite{gor_un2}.

\begin{corollary}
The determinacy problem for input-sized expressions is decidable.
This problem takes as input an input-sized expression 
$(\scm,e,\sigma)$ and asks if for
every instance $I$ conforming to
$\scm$ by $\sigma$, there exists at most one $B \in e(I)$.
\end{corollary}

Corollary~\ref{upperbound} gives an $\exists \mathbf R$ upper
bound on the combined complexity of query evaluation
\cite{vardi_comp}.  Our final result is a matching lower bound,
already for data complexity alone.  The proof is in the Appendix.

\begin{theorem} \label{hardness}
There exists a fixed schema $\scm$ and a fixed expression $e$ in
$\ML + \eigen$,
well-typed over $\scm$, such that the following problem is hard
for $\exists \mathbf R$: Given an integer
instance $I$ over $\scm$, decide
whether the zero matrix is a possible result of $e(I)$.
The pointwise applications in $e$ use only simple functions
definable by quantifier-free formulas over the reals.
\end{theorem}

\begin{remark}[Complexity of deterministic expressions]
Our proof of
Theorem~\ref{hardness} relies on the nondeterminism of the
$\eigen$ operation.
Coming back to our remark on determinacy at the end of the
previous section, it is an interesting question for further
research to understand not only the expressive power but also
the complexity of the evaluation problem
for \emph{deterministic} $\ML + \eigen$ expressions.
\end{remark}

\section{Conclusion}

There is a commendable trend in contemporary database research to
leverage, and considerably extend, techniques from database query processing
and optimization, to support large-scale linear algebra
computations.  In principle, data scientists could then work directly
in SQL or related languages.
Still, some users will prefer to continue using the
matrix sublanguages they are more familiar with.  Supporting
these languages is also important so that existing code need not
be rewritten.

From the perspective of database theory,
it then becomes
relevant to understand the expressive power of these languages as
well as possible.  In this paper we have
proposed a framework for viewing matrix
manipulation from the point of view of expressive power of database query
languages.  Moreover, our results formally confirm that the basic
set of matrix operations offered by systems in practice,
formalized here in the language $\ML + \inv + \eigen$, really is
adequate for expressing a range of linear algebra techniques
and procedures.

In the paper we have already mentioned some intriguing questions
for further research.  Deep inexpressibility results have
been developed for logics with rank operators \cite{pakusa_phd}.
Although these results are mainly concerned with finite fields,
they might still provide valuable insight in our open questions.  Also,
we have not covered all standard constructs from linear algebra.
For instance, it may be worthwhile to extend our framework
with the operation of putting matrices in upper triangular form,
with the Gram-Schmidt procedure (which is now partly hidden in
the $\eigen$ operation),
and with the singular value decomposition.

Finally, we note that various authors have proposed to go beyond
matrices, introducing data models and algebra for tensors or
multidimensional arrays \cite{rusu_survey,kim_tensordb,sato_tensors}.  When
moving to more and more powerful and complicated languages,
however, it becomes less clear at what point we should simply
move all the way to full SQL, or extensions of SQL with recursion.

\section*{Acknowledgment}

We thank Bart Kuijpers for telling us about the complexity class
$\exists \mathbf R$.  We thank Lauri Hella and Wied Pakusa for
helpful discussions, and Christoph Berkholz and Anuj Dawar
for their help with the proof of Proposition~\ref{4clique}.

\bibliographystyle{plainurl}
\bibliography{database}

\appendix

\section*{Appendix}

\begin{proof}[Proof of Theorem~\ref{fo3}]
\newcommand{\all}{\mathit{all}}
It is known \cite{tarskigivant,marxvenema_multi}
that $\rm FO^3$ graph queries
can be expressed in the algebra
of binary relations with the operations $\all$, identity, union, set
difference, converse, and relational composition.
These operations are well
known, except perhaps for $\all$, which, on a graph instance $I$,
evaluates to the cartesian product of the active domain of $I$
with itself.
Identity evaluates to the identity relation on the active domain
of $I$.  Each of
these operations is easy to express in $\ML$.  For $\all$ we use
$\one(R) \cdot \one(R)^*$, where for $R$ we can take any relation variable
from the schema.  Identity is expressed as $\diag(\one(R))$.
Union $r \cup s$ is expressed as
$\Apply[x \lor y](r,s)$, and set difference $r - s$ as $\Apply[x
\land \neg y](r,s)$.  Converse is transpose.  Relational
composition $r \circ s$ is expressed as $\Apply[x>0](r \cdot s)$,
where ${x>0} = 1$ if $x$ is positive and $0$ otherwise.
\end{proof}

\paragraph*{The
relational calculus with aggregates.}  In this logic, we have
base variables and numerical variables.  Base variables can be
bound to base columns of relations, and compared for equality.
Numerical variables can be bound to numerical columns, and can
be equated to function applications and aggregations.  We will
not recall the syntax formally \cite{libkin_sql}.
The advantage of the relational calculus is that variables,
especially base variables, can be repeated and reused.  For example,
matrix multiplication $M \cdot N$ with $M$ of type $\alpha \times
\beta$ and $N$ of type $\beta \times \gamma$ can be expressed by
the formula $$ \varphi(i,j,z) \equiv z = \Sum k,x,y . (M(i,k,x)
\land N(k,j,y), x \times y). $$  Here, $i$, $j$ and $k$ are base
variables and $x$, $y$ and $z$ are numerical variables.  Only two
base variables, $i$ and $j$, are free; in the subformula
$M(i,k,x)$ only $i$ and $k$ are free, and in $N(k,j,y)$ only $k$
and $j$ are free.  So, if $M$ or $N$ had been a subexpression
involving matrix multiplication in turn, we could have reused one
of the three variables.  The other operations of $\ML$ need
only two base variables.  We conclude:

\begin{proposition} \label{calc3}
Let $\scm$, $e$, $\tau$ and $\ell$ as in
Theorem~\ref{sumtheorem}.
For every $\ML$ expression $e$ there is a formula
$\varphi$ over $\Rel(\scm)$
in the relational calculus with summation, such that
\begin{enumerate}
\item
If\/ $\tau$ is general, $\varphi(i,j,z)$ has two free base
variables $i$ and $j$ and one free numerical variable $z$; if\/
$\tau$ is a vector type, we have $\varphi(i,z)$; and if\/ $\tau$ is
scalar, we have $\varphi(z)$.
\item
For every instance $I$, the relation defined by $\varphi$ on
$\Rel(I)$ equals $\Rel_\ell(e(I))$.
\item
The formula $\varphi$ uses only three distinct base variables.
\end{enumerate}
\end{proposition}

To prove Proposition~\ref{4clique} we state a lemma, which
refines Proposition~\ref{calc3} in the setting of graph queries.

\begin{lemma} \label{lemmapsi}
If a graph query $q$ is expressible in $\ML$, then $q$ is expressible
by a formula $\psi(i,j)$ in the relational calculus with
summation, where $i$ and $j$ are base
variables, and $\psi$ uses at most three distinct base
variables.
\end{lemma}
\begin{proof}
Let $e$ be a $\ML$ expression that expresses $q$.  Let
$\varphi(i,j,z)$
be the formula given by Proposition~\ref{calc3}.
Let $\varphi'(i,j,z)$ be the formula obtained from $\varphi$ as follows.
We replace each atomic formula of the form $R(i',j',x)$, where
$i'$ and $j'$ are base variables and $x$ is a numerical variable,
by $((x=1 \land R(i',j')) \lor (x=0 \land \neg R(i',j'))$.  Now
$\psi$ can be obtained as $\exists z(z=1 \land \varphi')$.
\end{proof}

We can now give the

\begin{proof}[Proof of Proposition~\ref{4clique}]
Let $e$ be a $\ML$ expression expressing some graph query $q$.
Let $\psi$ be the formula given by
Lemma~\ref{lemmapsi}.  It is known \cite{hlnw_aggregate,libkin_sql} that
every formula in the relational calculus with
aggregates can be equivalently expressed by a formula
in infinitary logic with counting, where the only variables in the latter
formula are the base variables in the original formula.  Hence,
$q$ is expressible in $C^3$, infinitary counting logic with three
distinct variables.

The four-clique query, however, is not expressible in $C^3$.  In
proof, consider the four-clique graph $G$, to which
we apply the Cai-F\"urer-Immerman construction
\cite{cfi,otto_bounded}, yielding graphs $G^0$ and $G^1$ which
are indistinguishable in $C^3$.  This construction is such that
$G^0$ contains a ``four-clique formed by paths of length three'':
four nodes such that there is a path of length three between any
two of them.  The graph $G^1$, however, does not contain four
such nodes.

Now suppose, for the sake of
contradiction, that there would be a sentence $\varphi$ in $C^3$
expressing the existence of a four-clique.  We can replace each
atomic formula $R(x,y)$ by $\exists z(R(x,z) \land \exists
x(R(z,x) \land R(x,y)))$.  The resulting $C^3$ sentence looks
for a four-clique formed by paths of length three, and would
distinguish $G^0$ from $G^1$, which yields our contradiction.
\end{proof}

\begin{proof}[Proof of Theorem~\ref{inv2eigen}]
We describe a fixed procedure for determining $A^{-1}$,
for any square matrix $A$.  Let $S = A^*A$.  Then $A$ is
invertible if and only if $S$ is.  Let us assume first that $S$
is indeed invertible.

Since $S$ is self-adjoint, $\C^n$ has an orthogonal basis
consisting of eigenvectors of $S$.  Eigenvectors of a
self-adjoint operator that correspond to distinct eigenvalues are
always orthogonal.  Hence, $\eigen(S)$ always returns an
orthogonal basis of $\C^n$ consisting of eigenvectors of $S$.
Let $(B,\Lambda) \in \eigen(S)$ (using the shorthand introduced
before Example~\ref{exrank}).  We can normalize the columns of
$B$ in $\ML$ as $$ \Apply[x/\sqrt y](B,\one(B) \cdot (B^*\cdot B
\cdot \one(B))^*).$$  (This expression works because the columns
in $B$ are mutually orthogonal.)
So, we may now assume that $B$ contains an
orthonormal basis consisting of eigenvectors of $S$.  In
particular, $B^{-1}=B^*$, and $S = B \Lambda B^*$.

\newcommand{\LL}{\sqrt{\Lambda}} Since we have assumed $S$ to be
invertible, none of the eigenvalues is zero.  We can invert
$\Lambda$ simply by replacing each entry on the diagonal by its
reciprocal.  Thus, $\Lambda^{-1}$ can be computed from $\Lambda$
by pointwise application.

Now $A^{-1}$ can be computed by the expression $C=B \Lambda^{-1} B^* A^*$.
To see that $C$ indeed equals $A^{-1}$, we calculate $CA =
B \Lambda^{-1} B^* A^* A = B \Lambda^{-1} B^* S = B \Lambda^{-1}
B^* B \Lambda B^*$ which simplifies to the identity matrix.

When $S$ is not invertible, we should return the zero matrix.
In $\ML$ we can compute the  matrix $Z$ that is zero if one of the
eigenvalues is zero, and the identity matrix otherwise.  We then
multiply the final expression with $Z$.
A final detail is to make the computation
well-defined in all cases.  Thereto, in the pointwise
applications of $x/\sqrt y$ and the reciprocal, we extend these
functions arbitrarily to total functions.
\end{proof}

\begin{proof}[Proof of Theorem~\ref{hardness}]
The feasibility problem \cite{schaefer_nash}
takes as input an equation $p=0$, with $p$ a
multivariate polynomial with integer
coefficients, and asks whether the equation has a solution over
the reals.  We may assume that $p$ is given in ``standard form'', as a sum of
terms of the form $a\mu$ where $a$ is an integer and $\mu$ is a
monomial \cite{matousek_existsr}.  The feasibility problem is
known to be complete for $\exists \mathbf R$. We will design a
schema $\scm$ and an expression $e$ so that the feasibility
problem reduces in polynomial time to our problem.

We use a construction by Valiant \cite{valiant_algebra}.  This
construction converts any $p$ as above, in polynomial time, to a
directed, edge-weighted graph $G$.  The fundamental property of
Valiant's construction is that the determinant of the adjacency
matrix $A$ of $G$ equals $p$.  The edge weights in $G$
are coefficients or variables from $p$, or the value 1.  The
entries in $A$ are zero or edge weights from $G$.  We now observe
that the construction has a specific property: when $p$ is given
in standard form, with an explicit coefficient before each
monomial (even if it is merely the value 1), each row of $A$
contains at most one variable.  This property is important for
the expression $e$, specified below, to work.

\newcommand{\Coef}{\mathit{Coef}}
\newcommand{\Vars}{\mathit{Vars}}
\newcommand{\Enc}{\mathit{Enc}}
Assume $G$ has nodes $1$, \dots, $n$, and let the variables
in $p$ be $x_1$, \dots, $x_k$.  We represent $A$ by three integer
matrices $\Coef$, $\Vars$, and $\Enc$.  Matrix $\Coef$ is the $n \times n$
matrix obtained from $A$ by omitting the variable entries (these
are set to zero).  On the other hand, $\Vars$, also $n \times n$, is
obtained from $A$ by keeping only the variable entries, but
setting them to 1.  All other entries are set to zero.  Finally,
$\Enc$ encodes which variables are represented by the one-entries in
$\Vars$.  Specifically, $\Enc$ is the $n \times k$ matrix where
$E_{i,j}=1$ if the $i$th row of $A$ contains variable $x_j$, and
zero otherwise.

We thus reduce an input $p=0$ of the feasibility problem to the
instance $I$ consisting of the matrices $\Coef$, $\Vars$, $\Enc$.
Additionally, for technical reasons, $I$ also has the $k \times
1$ column vector $F$, which has value 1 in its first entry and is
zero everywhere else.  Formally, this instance is over the fixed
schema $\scm$ consisting of the matrix variables $M_{\Coef}$,
$M_{\Vars}$, $M_{\Enc}$, and $M_F$, where the first two variables
have type $\alpha \times \alpha$; the third variable has type
$\alpha \times \beta$; and $M_F$ has type $\beta \times 1$.  To
reduce clutter, however, in what follows we will write these
variables simply as $\Coef$, $\Vars$, $\Enc$ and $F$.

We must now give a expression $e$ that has the zero matrix
as possible result of $e(I)$ if and only if $p=0$ has a solution over the
reals.  For any $k \times 1$ vector $v$ of real numbers, let
$A^{(v)}$ denote the matrix $A$ where we have substituted the
entries of $v$ for the variables $x_1$, \dots, $x_k$.  By
Valiant's construction, the
expression $e$ should return the zero matrix as a possible
result, if and only if there exists $v$ such that $A^{(v)}$ has 
determinant zero, i.e., is not invertible.

The desired expression $e$ works as follows.
By applying $\eigen$ to the $k\times k$ zero matrix $O$,
and selecting the first column, we can nondeterministically
obtain all possible nonzero $k \times 1$ column vectors.
Taking only the real part ($\Re$) of the entries,
we obtain all possible real column vectors $v$.
Then the matrix $A^{(v)}$ is assembled (in matrix variable $AA$)
using the matrices
$\Coef$, $\Vars$ and $\Enc$.  Finally, we apply $\inv$ to $AA$ so that
the zero matrix is returned if and only if $AA$ has determinant zero.

In conclusion, expression $e$ reads as follows:
\begin{tabbing}
\sffamily let $O = \Apply[0](F \cdot F^*)$ in \\
\sffamily let $B = \eigen(O)$ in \\
\sffamily let $v = \Apply[\Re](B \cdot F)$ in \\
\sffamily let $AA = \Apply[+](\Coef,\Apply[g](\Vars,\Enc\cdot v
\cdot \one(\Coef)^*))$ in \\
$\inv(AA)$
\end{tabbing}
Here, in the last expression, $g(x,y)=y$ if $x=1$, and zero
otherwise.
\end{proof}

\end{document}